\documentclass[a4paper,11pt]{article}
\usepackage[margin=1.2in]{geometry}

\renewenvironment{abstract}
 {\small
  \vspace{-1em}
  \begin{center}
  \bfseries \abstractname\vspace{-.5em}\vspace{0pt}
  \end{center}
  \list{}{
    \setlength{\leftmargin}{0.6in}%
    \setlength{\rightmargin}{\leftmargin}%
  }%
  \item\relax}
 {\endlist}

\usepackage[utf8]{inputenc}
\usepackage[T1]{fontenc}
\usepackage[textwidth=2.5cm, color=yellow]{todonotes}
\usepackage{amsmath, amsthm, amssymb}
\usepackage{graphicx}
\usepackage[colorlinks=true, citecolor=red]{hyperref}
\usepackage[ruled, linesnumbered, vlined, algo2e]{algorithm2e}
\usepackage[noadjust]{cite}
\usepackage{needspace}

\usepackage{aliascnt}
\newtheorem{theorem}{Theorem}[section]

\newaliascnt{lemma}{theorem}
\newtheorem{lemma}[lemma]{Lemma}
\aliascntresetthe{lemma}

\newaliascnt{proposition}{theorem}
\newtheorem{proposition}[proposition]{Proposition}
\aliascntresetthe{proposition}

\newaliascnt{definition}{theorem}

\aliascntresetthe{definition}

\newaliascnt{corollary}{theorem}
\newtheorem{corollary}[corollary]{Corollary}
\aliascntresetthe{corollary}

\newaliascnt{conjecture}{theorem}

\aliascntresetthe{conjecture}

\theoremstyle{remark}
\newaliascnt{claim}{theorem}

\aliascntresetthe{claim}

\newaliascnt{observation}{theorem}

\aliascntresetthe{observation}

\newaliascnt{example}{theorem}
\newtheorem{example}[example]{Example}
\aliascntresetthe{example}

\newaliascnt{remark}{theorem}
\newtheorem{remark}[remark]{Remark}
\aliascntresetthe{remark}

\newaliascnt{problem}{theorem}
\newtheorem{problem}[problem]{Open problem}
\aliascntresetthe{problem}

\newcommand{\cqedsymbol}{\ifmmode$\lrcorner$\else{\unskip\nobreak\hfil
\penalty50\hskip1em\null\nobreak\hfil$\lrcorner$
\parfillskip=0pt\finalhyphendemerits=0\endgraf}\fi}

\usepackage{tabularx}
\usepackage{environ}
\makeatletter
\newcommand{\problemtitle}[1]{\gdef\@problemtitle{#1}}
\newcommand{\probleminput}[1]{\gdef\@probleminput{#1}}
\newcommand{\problemquestion}[1]{\gdef\@problemquestion{#1}}
\NewEnviron{decproblem}{
  \problemtitle{}\probleminput{}\problemquestion{}
  \BODY
  \par\addvspace{.5\baselineskip}
  \noindent
  \begin{tabularx}{\textwidth}{@{\hspace{\parindent}} l X c}
    \multicolumn{2}{@{\hspace{\parindent}}l}{\@problemtitle} \\
    \textbf{Input:} & \@probleminput \\
    \textbf{Question:} & \@problemquestion
  \end{tabularx}
  \par\addvspace{.5\baselineskip}}
\NewEnviron{genproblem}{
  \problemtitle{}\probleminput{}\problemquestion{}
  \BODY
  \par\addvspace{.5\baselineskip}
  \noindent
  \begin{tabularx}{\textwidth}{@{\hspace{\parindent}} l X c}
    \multicolumn{2}{@{\hspace{\parindent}}l}{\@problemtitle} \\
    \textbf{Input:} & \@probleminput \\
    \textbf{Output:~~} & \@problemquestion
  \end{tabularx}
  \par\addvspace{.5\baselineskip}}

\newcommand{\I}{\mathcal{I}} 
\newcommand{\B}{\mathcal{B}} 
\newcommand{\C}{\mathcal{C}} 
\newcommand{\G}{\mathcal{G}} 
\renewcommand{\H}{\mathcal{H}} 
\newcommand{\E}{\mathcal{E}} 
\renewcommand{\L}{\mathcal{L}} 
\newcommand{\M}{\mathcal{M}} 
\newcommand{\J}{\mathcal{J}} 
\newcommand{\Sig}{(X,\Sigma)} 

\newcommand\myrot[1]{\mathrel{\rotatebox[origin=c]{#1}{$\Rightarrow$}}}
\newcommand\NEarrow{\myrot{45}}

\DeclareMathOperator{\nn}{\small\NEarrow}

\DeclareMathOperator{\Min}{Min}
\DeclareMathOperator{\Max}{Max}

\DeclareMathOperator{\pred}{pred}

\usepackage{old-arrows}
\DeclareMathOperator{\persp}{\mathrel{\text{\ooalign{$\swarrow$\cr$\nearrow$}}}}

\definecolor{belize}{RGB}{0, 143, 216}
\definecolor{teal}{RGB}{0, 180, 166}

\begin{document}

\title{Translating between the representations \\ of a ranked convex geometry\thanks{The first two authors have been supported by the ANR project GraphEn ANR-15-CE40-0009. The last author is funded by the CNRS, France ProFan project.}}

\date{March 30, 2021}

\author{
    Oscar Defrain\thanks{LIMOS, Université Clermont Auvergne, France.}
    \addtocounter{footnote}{-1} 
    \and
    Lhouari Nourine\footnotemark
    \and
    \addtocounter{footnote}{-1} 
    Simon Vilmin\footnotemark
}

\maketitle

\begin{abstract}
    It is well known that every closure system can be represented by an implicational base, or by the set of its meet-irreducible elements.
    In Horn logic, these are respectively known as the Horn expressions and the characteristic models.
    In this paper, we consider the problem of translating between the two representations in acyclic convex geometries.
    Quite surprisingly, we show that the problem in this context is already harder than the dualization in distributive lattices, a generalization of the well-known hypergraph dualization problem for which the existence of an output quasi-polynomial time algorithm is open.
    In light of this result, we consider a proper subclass of acyclic convex geometries, namely ranked convex geometries, as those that admit a ranked implicational base analogous to that of ranked posets. 
    For this class, we provide output quasi-polynomial time algorithms based on hypergraph dualization for translating between the two representations.
    
    \vskip5pt\noindent{}{\bf Keywords:} hypergraph dualization, meet-irreducible enumeration, characteristic models, implicational bases, convex geometries, lattices.
\end{abstract}

\section{Introduction}\label{sec:introduction}

Finite closure systems arise in various fields of discrete mathematics and computer science including Horn logic~\cite{kautz1993reasoning,crama2011boolean}, relational databases~\cite{maier1983theory,mannila1992design}, lattice theory~\cite{davey2introduction,caspard2003lattices}, Formal Concept Analysis (FCA)~\cite{ganter2012formal}  and knowledge spaces~\cite{doignon2012knowledge}. 
The study of their representations and how to translate from one to another has gathered increasing attention these last decades, as witnessed by~\cite{wild1994theory,khardon1995translating,beaudou2017algorithms,habib2018representation} and the Dagstuhl Seminar 14201~\cite{dagstuhl2014}.
See~\cite{wild2017joy,rudolph2017succinctness} for recent surveys on these topics.

Among the different ways of representing a closure system, the {\em implicational bases} play a central role.
Essentially, they consist of rules $A \rightarrow B$ describing a causality relation within the closure system: a set containing $A$ must contain $B$.
If every conclusion $B$ has size one then the implicational base is called unit.
Several implicational bases can lead to the same closure system, and consequently numerous bases were defined.
Among the most frequent, one can find the Duquenne-Guigues basis having a minimum number of implications~\cite{guigues1986familles}, the unit-minimum having a minimum number of implications among unit implicational bases, or the canonical direct basis having all minimal generators~\cite{bertet2010multiple}. 
Refinements of the canonical direct basis include the $E$-basis and $D$-basis~\cite{freese1995free,adaricheva2014implicational,adaricheva2017discovery}.
Another possible representation for a closure system results from a minimum subset of elements from which it can be reconstructed. 
In Horn logic, these elements are known as the \textit{characteristic models}~\cite{khardon1995translating, hammer1995quasi}.
In lattice theory and closure systems, they are known as the \textit{meet-irreducible elements}~\cite{davey2introduction}. 
They can be found in the {\em poset of irreducibles}~\cite{barbut1970ordre,markowsky1975factorization,habib2018representation}, or in the \textit{reduced context} representing the concept lattice in FCA~\cite{ganter2012formal}.

As pointed out by Khardon in~\cite{khardon1995translating}, the utility of these two representations is not comparable.
More notably, there are cases of (unit) minimum implicational bases of exponential size in the number of meet-irreducible elements, and vice versa.
Furthermore, different complexities can arise for a same problem when considering one representation, and the other.
This is in particular the case for reasoning, abduction, or when considering dualization problems in lattices~\cite{kautz1993reasoning,babin2017dualization,defrain2019dualization}.
Hence, the problem of translating between implicational bases and meet-irreducible elements is critical in order to reap the benefits of both representations.
In Horn logic, the problem of computing the meet-irreducible elements from an implicational base is known as \textsc{CCM} (for \emph{Computing Characteristic Models}).
The problem of computing an implicational base from the meet-irreducible elements is denoted by \textsc{SID} (for \emph{Structure Identification}).
For the canonical direct and $D$-basis, output quasi-polynomial time algorithms based on hypergraph dualization were given for both translations in~\cite{mannila1992design,adaricheva2017discovery}.
An enumeration algorithm is said to be running in {\em output-polynomial} time if its running time is bounded by a polynomial in the combined size of the input and the output.
It is said to be running with {\em polynomial delay} if after polynomial-time preprocessing, the running times between two consecutive outputs and after the last output are bounded by a polynomial in the size of the input only~\cite{johnson1988generating}.
For the minimum implicational base, output-polynomial time algorithms were obtained in $k$-meet-semidistributive lattices in~\cite{beaudou2017algorithms}, and in modular lattices in~\cite{wild2000optimal}. 
In~\cite{babin2013computing}, it is shown that it is co{\sf NP}-complete in general to decide whether an implication belongs to a minimum implicational base.
However, the existence of an output-polynomial or output quasi-polynomial time algorithm constructing a minimum implicational base remains open~\cite{babin2013computing,beaudou2017algorithms,wild2017joy}.
In~\cite{khardon1995translating}, the author is most interested in unit implicational bases as they represent Horn expressions.
In that case, the translation problems \text{SID} and \text{CCM} are shown to be equivalent, and to be harder than hypergraph dualization.
Moreover, computing a unit-minimum implicational base is harder than a minimum one, as the latter can be obtained in polynomial time from the former using the algorithms in~\cite{shock1986computing,wild1995computations}, and that it is only polynomially smaller in the size of the ground set.
Yet, the existence of an output-polynomial or output quasi-polynomial time algorithm solving the problem in both contexts remains open.

In this paper, we focus on unit implicational bases and acyclic convex geometries, i.e., closure systems satisfying the anti-exchange property~\cite{edelman1985theory} and admitting an acyclic implication-graph~\cite{wild1994theory}.
This class has been widely considered in the literature~\cite{hammer1995quasi,wild1994theory,boros2009subclass}.
Even when restricted to this class, we show that the problem of translating between one representation and the other is harder than the dualization in distributive lattices, a generalization of the hypergraph dualization problem.
For this problem, the existence of an output quasi-polynomial time algorithm is open \cite{babin2017dualization,defrain2019dualization}.
Not only this generalizes the observation of Khardon in~\cite{khardon1995translating}, but it surprisingly holds in a  very low class of closure spaces.
In light of this result, we then consider a proper subclass of acyclic convex geometries, namely ranked convex geometries, as those that admit a ranked implicational base analogous to that of ranked posets. 
For this class, we provide output quasi-polynomial time algorithms based on hypergraph dualization for translating between the two representations.
The two algorithms proceed as follows.
For the enumeration of the meet-irreducible elements (\textsc{CCM}), the first algorithm constructs solutions from the join-irreducible elements by listing all the maximal sets which do not imply them, rank by rank.
It relies on a recursive partition of the solutions conducted at each step of the algorithm.
As for the construction of a unit-minimum implicational base (\textsc{SID}), the second algorithm relies on a bijection between hypergraph transversals and the minimal generators that belong to the base.
If in addition the size of the premises in the implicational base, or the size of the meet-irreducible in the lattice are bounded by constants, then the algorithm is shown to perform in output-polynomial time using the algorithm of Eiter and Gottlob in~\cite{eiter1995identifying}.
These positive and negative results improve the understanding of a long-standing open problem.

The rest of the paper is structured as follows. 
In Section~\ref{sec:preliminaries} we introduce concepts and notations that will be used throughout the paper.
In Section~\ref{sec:hardness} we give hardness results in acyclic convex geometries in relation with the dualization problem in distributive lattices. 
Section~\ref{sec:ranked} presents ranked convex geometries, while Section~\ref{sec:mis} and~\ref{sec:tr} are devoted to the aforementioned algorithms.
We conclude with perspectives in Section~\ref{sec:conclusion}.

\section{Preliminaries}\label{sec:preliminaries}

All the objects considered in this paper are finite.
For a set $X$ we denote by $2^X$ the set of all subsets of $X$.

A {\em partial order} $P$ on a set $X$ (or {\em poset}) is a binary relation $\leq$ on $X$ which is reflexive, anti-symmetric and transitive, denoted by $P=(X,\leq)$.
Two elements $x$ and $y$ of $P$ are said {\em comparable} if $x \leq y$ or $y \leq x$, and {\em incomparable} otherwise.
We~note $x<y$ if $x\leq y$ and $x\neq y$.
We say that $x$ \emph{covers} $y$ ($y$ is covered by $x$) and denote $y \prec x$ if $y \leq x$ and there
is no $z\in X$ such that $x<z<y$.
In this case, we also say that $x$ is a \emph{successor} of $y$, and $y$ a \emph{predecessor} of $x$.
A~subset of a poset in which no two distinct elements are comparable is called an {\em antichain}.
If $A\subseteq X$, then $P[A]$ denotes the order induced by the elements of~$A$, and $P-A$ the poset $P[X\setminus A]$.
An {\em upper bound} of $x$ and $y$ is an element $u \in X$ such that $x \leq u$ and $y \leq u$. 
A {\em lower bound} for $x$ and $y$ will satisfy the dual statement.
If moreover $u$ is minimum among all upper bounds of $x$ and $y$, then it is called {\em least upper bound} (also known as {\em supremum} or {\em join}) of $x$ and $y$ and is denoted by $x \lor y$. 
The {\em greatest lower bound} $x \land y$ of $x$ and $y$ (or {\em infimum}, {\em meet}) is defined dually. 
Note that $x \lor y$ and $x \land y$ may or may not exist in general.
A {\em lattice} is a poset in which every two elements have an infimum and a supremum; see~\cite{birkhoff1940lattice,davey2introduction}. 
It has in particular a maximum element, the \emph{top}, and a minimum one, the \emph{bottom}.
An element covering the bottom is an \emph{atom}, while an element covered by the top is a \emph{co-atom}.
A lattice is called {\em Boolean} if it is isomorphic to $(2^X,\subseteq)$ for some set $X$.
It is called {\em distributive} if for any three elements $x,y,z$ of the lattice,
\[
  x \wedge (y \vee z)=(x \wedge y) \vee (x \wedge  z).
\]
For $x \in X$, we call {\em principal ideal} (or just {\em ideal}) of $x$ the set $\downarrow x = \{ y \in X \mid y \leq x \}$.
Analogously, $\uparrow x = \{ y \in X \mid y \geq x \}$ is the {\em (principal) filter} of $x$.
More generally, we set the ideal $\downarrow S$ of $S \subseteq X$ as $\downarrow S=\bigcup_{x\in S} \downarrow x$ and its filter accordingly.

We define notions related to closure systems.
A map $\phi : 2^X\rightarrow 2^X$ is a {\em closure operator} on $X$ if for all $A,B\subseteq X$:
\begin{enumerate}
    \item $A \subseteq \phi(A)$;
    \item $A \subseteq B$ implies $\phi(A) \subseteq \phi(B)$; and
    \item $\phi(A) = \phi(\phi(A))$.
\end{enumerate}
A {\em closed set} of $X$ w.r.t.~$\phi$ is a set $C \subseteq X$ such that $\phi(C)=C$.
The set of all closed sets of $X$ w.r.t.~$\phi$ is denoted by $\C_{\phi}$.
A pair $(X, \phi)$ where $\phi$ is a closure operator on $X$ is called {\em closure space}.
It is called {\em standard} if moreover
\begin{enumerate}
    \item $\phi(\emptyset) = \emptyset$; and
    \item $\phi(x)\setminus \{ x \}$ is closed for all $x \in X$.
\end{enumerate}
In this paper, all closure spaces are considered standard, a common assumption~\cite{gratzer2016lattice}.
A {\em closure system} is a pair $(X, \C)$ where $\C \subseteq 2^X$, $X \in \C$ and $C_1 \cap C_2 \in \C$ for all $C_1, C_2 \in \C$.
It is well known that to every closure space $(X,\phi)$ corresponds a closure system $(X, \C_\phi)$, and that to every closure system $(X,\C)$ corresponds a closure space $(X,\phi)$ where $\C = \C_{\phi}$.
Furthermore, $\L_\phi=(\C_{\phi},\subseteq)$ is a lattice whenever $(X,\phi)$ is a closure space~\cite{birkhoff1940lattice, davey2introduction}. 
Then $C_1\wedge C_2 = C_1\cap C_2$ and $C_1\vee C_2 = \phi(C_1\cup C_2)$.
The lattice of a closure space is given in Figure~\ref{fig:lattice}.

We define meet-irreducible elements.
Let $(X,\C=\C_\phi)$ be a closure system with closure operator $\phi$.
A set $M \subseteq X$ is a {\em meet-irreducible element} of $\C$ if $M\in \C$, $M\neq X$, and for all $C_1, C_2 \in \C$ such that $M=C_1\cap C_2$ it follows that $C_1 = M$ or $C_2 = M$. 
Similarly, a set $J \subseteq X$ is a {\em join-irreducible element} of $\C$ if $J\in \C$, $J\neq \emptyset$ and for all $C_1, C_2 \in \C$ such that $J=\phi(C_1 \cup C_2)$ it follows that $J = C_1$ or $J = C_2$ .
We denote by $\J(\C)$ and $\M(\C)$ the sets of join-irreducible and meet-irreducible elements of $\C$.
In the lattice $\L=(\C,\subseteq)$, $\J(\C)$ and $\M(\C)$ respectively correspond to the elements that have a unique predecessor and a unique successor; see Figure~\ref{fig:lattice} for an example.
In the following, we will interchangeably note $\M(\L)$ and $\M(\C)$ whenever $\L=(\C,\subseteq)$.
Let $J,M\in \C$.
We define useful notations from~\cite{ganter2012formal} for closure systems using underlying lattice structure.
We~note $J\nearrow M$ whenever $M$ is maximal in $\L-\!\!\uparrow\!J$, and $J \swarrow M$ whenever $J$ is minimal in $\L-\!\!\downarrow\!M$.
If $J\nearrow M$ and $J\swarrow M$ then we note $J \persp M$.
We point out here that $M$ and $J$ are element of $\L$, hence subsets of $X$.
On the example of Figure~\ref{fig:lattice}, observe that $J\nearrow M$ for $J=\{2\}$ and $M=\{1,4\}$.
Also, $J \swarrow M$ for $M=\{1,4\}$ and $J=\{2\}$.
Hence, $J \persp M$ in that case.
In the following, we extend the $\nearrow$ notation to any subset $B\subseteq X$ by defining $B^\nearrow=\Max_\subseteq\{M \in \C \mid B\not\subseteq M\}$.
Then $M\in J^\nearrow$ whenever $J\nearrow M$ for any two $M,J\in \C$.
We put $j^\nearrow=\{j\}^\nearrow$ for any $j\in X$.
It~is well known that the family of sets $\{J^\nearrow \mid J\in \J(\C)\}$ is a (possibly overlapping) set covering of $\M(\C)$. 
See for instance~\cite{mannila1992design,wild1995computations,wild2017joy}.
Since $\C$ is considered standard, another well-known property is that to every $J\in \J(\C)$ corresponds a unique $j\in X$ such that $J=\phi(j)$.
Accordingly, $\J(\C)$ coincides with $\{\phi(j) \mid j \in X \}$, and $J^\nearrow$ with $j^\nearrow$ for $j$ such that $\phi(j)=J$.

We now turn to definitions on implications, and refer to~\cite{wild2017joy,bertet2018lattices} for recent surveys on this topic.
An implicational base $\Sig$ is a set $\Sigma$ of implications of the form $A \rightarrow B$ where $A\subseteq X$ and $B\subseteq X$. 
In this paper we only consider $\Sigma$ in its equivalent {\em unit} form where $|B|=1$ for every implication, and denote by $A\rightarrow b$ such implications.
We call {\em premise} of $A\rightarrow b$ the set $A$, and $b$ its {\em conclusion}.
We call {\em size} of $\Sigma$ and note $|\Sigma|$ the number of implications in~$\Sigma$.
We call {\em dimension} of $\Sigma$ the size of the largest premise in~$\Sigma$.

%
A~set $S$ is {\em closed} in $\Sigma$ if for every implication $A\rightarrow b$ of $\Sigma$, at least one of $b\in S$ and $A\not\subseteq S$ holds.
To~$\Sigma$ we associate the {\em closure operator} $\phi$ which maps every subset $S\subseteq X$ to the smallest closed set  $\phi(S)$ of $\Sigma$ containing~$S$.
Then we say that $S$ implies $x\in X$ if $x\in \phi(S)$.
We will interchangeably denote by $\C_{\Sigma}$ and $\C_{\phi}$ the set of all closed sets of~$\Sigma$.
Observe that $\Sig$ defines a closure space, hence that $(X,\C_{\Sigma})$ defines a closure system and $\L_\Sigma=(\C_\Sigma,\subseteq)$ a lattice.
Note that to a single closure system can correspond several implicational bases.
We say that two implicational bases $\Sigma$ and $\Sigma'$ are equivalent, denoted by $\Sigma \equiv \Sigma'$, if $\C_\Sigma=\C_{\Sigma'}$.
An implicational base $\Sigma$ is called {\em irredundant} if $\Sigma\setminus \{ A\rightarrow b \}\not\equiv \Sigma$ for all $A\rightarrow b\in \Sigma$. 
An implicational base is {\em unit-minimum} if it is of minimum size among all equivalent unit implicational bases.
A set $A\subseteq X$ is a {\em minimal generator} of $b$ if $b\in \phi(A)$ and $b\not\in \phi(A\setminus \{x\})$ for any $x\in A$.
The {\em implication-graph} of $\Sig$ is the directed graph $G(\Sigma)$ defined on vertex set $X$ and where there is an arc between $x$ and $y$ if there exists $A\rightarrow b\in \Sigma$ such that $x\in A$ and $y=b$.
An implicational base is called {\em acyclic} if its implication-graph has no directed cycle.
A closure space is called {\em acyclic} if it admits an acyclic implicational base.
If moreover the premises in $\Sigma$ are of size one, then the closure lattice $\L_\Sigma=(\C_\Sigma,\subseteq)$ is distributive, and this is in fact a characterization~\cite{birkhoff1940lattice,davey2introduction}.
Throughout the paper, for a vertex $j$ of $G(\Sigma)$ we shall note $j^-$ the set of all predecessors of~$j$ in $G(\Sigma)$, and $j^+$ all its successors.

Observe that if $(X, \C_\phi)$ is a closure system, then it can be represented and reconstructed from a well-chosen implicational base, or from its meet-irreducible elements.
For the latter case, one has to close $\M(\C_{\phi})$ under intersection to recover the whole family $\C_{\phi}$. 
For the former one, we have to consider $\Sigma$ such that $\C_{\Sigma} = \C_{\phi}$. 
Such a $\Sigma$ always exists: $\Sigma = \{ A \rightarrow \phi(A) \mid A \subseteq X \}$ is an expensive but sufficient representation of $\phi$.
It can furthermore be turned into a unit implicational base.
An example of a closure system represented by its corresponding lattice, a unit-minimum implicational base, and the set of its meet-irreducible elements is given in Figure~\ref{fig:lattice}.

\begin{figure}
  \center
  \includegraphics[scale=1.1]{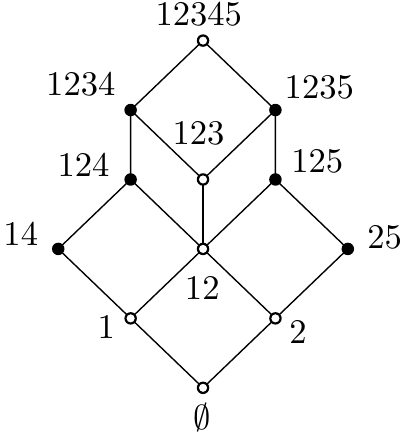}
  \caption{The closure lattice $\L_{\Sigma} = (\C_{\Sigma}, \subseteq)$ of the implicational base $\Sigma=\{4 \rightarrow 1, 5 \rightarrow 2,\allowbreak{} {3 \rightarrow 1}, 3 \rightarrow 2, 45 \rightarrow 3\}$ 
  on ground set $X=\{1,2,3,4, 5\}$. 
  Meet-irreducible elements are represented by black vertices.
  For better readability, closed sets are denoted without braces in the lattice, i.e., 
  $123$ stands for $\{1,2,3\}$.}\label{fig:lattice}
\end{figure}

Before diving into special classes of closure systems, let us include the definitions of the problems we consider in this paper.
As specified in the introduction, the names \textsc{CMI}, \textsc{CCM} and \textsc{SID} below originally come from the problems of translating between Horn representations and their characteristic models~\cite{khardon1995translating}: they stand for {\em Characteristic Models Identification}, {\em Computing Characteristic Models} and {\em Structure Identification}.
Their equivalence with the problems of translating between implicational bases and meet-irreducible elements is well established; see for instance~\cite{beaudou2017algorithms,wild2017joy}.
Note that the last problem calls for constructing a unit-minimum implicational base.
This problem is harder than the one of computing a minimum implicational base, as one can compute a minimum implicational base in polynomial time from a unit-minimum one using the algorithms in~\cite{shock1986computing,wild1995computations}, and that the size of these two implicational bases only differ by a factor $|X|$ being part of the input.
It is not clear whether the opposite direction holds, as it was shown in~\cite{hammer1993optimal} that deciding whether a unit implicational base can be reduced in size is {\sf NP}-hard.

\begin{decproblem}
  \problemtitle{Meet-irreducible elements identification (\textsc{CMI})}
  \probleminput{An implicational base $\Sig$ and a family of sets $\M\subseteq 2^X$.}
  \problemquestion{Is $\M=\M(\C_{\Sigma})$?}
\end{decproblem}

\begin{genproblem}
  \problemtitle{Meet-irreducible elements enumeration (\textsc{CCM})}
  \probleminput{An implicational base $\Sig$.}
  \problemquestion{The set $\M(\C_{\Sigma})$.}
\end{genproblem}

\begin{genproblem}
  \problemtitle{Implicational base identification (\textsc{SID})}
  \probleminput{Two sets $X$ and $\M\subseteq 2^X$.}
  \problemquestion{A unit-minimum implicational base $\Sig$ such that $\M=\M(\C_{\Sigma})$.}
\end{genproblem}

We point that \textsc{CMI} is a decision problem while \textsc{CCM} and \textsc{SID} are generation problems.
As $\{j^\nearrow \mid j\in X\}$ is a set covering of $\M(\C)$, it is well known that the existence of an output-polynomial time algorithm for \textsc{CCM} can be reduced to the existence of one enumerating $j^\nearrow$ for all $j\in X$; see~\cite{wild1995computations,wild2017joy}.
Repetitions are either avoided using exponential memory, or by running the algorithm again on each output to check whether the solution has already been outputted before, in the fashion of~\cite[Lemma 5.1]{bonamy2019enumeratingkt} and at the cost of an increasing complexity.
In the general case, it can be seen using a result of Babin and Kuznetsov, and of Kavvadias et al.~in~\cite{kavvadias2000generating,babin2017dualization} on the intractability of generating the co-atoms of a lattice that the computation of $j^\nearrow$ is impossible in output-polynomial time unless {\sf P$=$NP}.
It is however a long-standing open problem whether such a result can be inferred for \textsc{CCM}~\cite{khardon1995translating,beaudou2017algorithms,wild2017joy}.

Let us now consider particular closure systems. 
A~closure space $(X,\phi)$ satisfies the {\em anti-exchange property} if for all $x\neq y$ and all closed sets $A\subseteq X$,
\[
    x\in \phi(A\cup \{y\})\ \text{and}\ x\not\in A\ \text{imply}\ y\not\in \phi(A\cup \{x\}).
\]
A standard closure space that satisfies the anti-exchange property is called a {\em convex geometry}.
Convex geometries are known to include acyclic closure spaces, also known as \emph{poset type} closure operators~\cite{wild1994theory,santocanale2014lattices, wild2017joy}.
Sometimes during the paper, we will refer to acyclic closure spaces as {\em acyclic convex geometries} (or \emph{ACG} in short).
It is known from~\cite{adaricheva2003join} that lattices of convex geometries are both join-semidistributive and lower-semimodular. 
As a consequence, if $(X,\phi)$ is a convex geometry and $(X, \C_{\phi})$ is its associated closure system, then $J\nearrow M$ implies $J\persp M$ for all $J\in \J(\C_{\phi})$ and $M\in \M(\C_{\phi})$~\cite[Corollary~4.6.3]{stern1999semimodular}, and the set $\{j^\nearrow \mid j\in X\}$ defines a partition of~$\M(\C_{\phi})$.
Consequently, meet-irreducible elements can be enumerated from join-irreducible elements with no need of handling repetitions in that case.
This yields the next proposition (the factor two comes from the delay between the output of the last element in $j^\nearrow$, $j\in X$ and the first element in $j'^\nearrow$ for two consecutive $j,j'\in X$, $j'\neq j$).

\begin{proposition}\label{prop:meets-reduces-to-colors}
    Let $f:\mathbb{N}\rightarrow\mathbb{N}$ be a function and $(X,\C_{\Sigma})$ be the closure system of a given acyclic implicational base $\Sig$.
    Then there is an algorithm enumerating the set $\M(\C_{\Sigma})$ of meet-irreducible elements of $\C_{\Sigma}$ with delay at most
    $2\cdot f(|X|+|\Sigma|)$
    whenever there is one enumerating $j^\nearrow$ with delay $f(|X|+|\Sigma|)$ given any $j\in X$.
\end{proposition}

We end the preliminaries with a few notions from hypergraph theory.
A {\em hypergraph} $\H$ is a couple $(V(\H),\E(\H))$ where $V(\H)$ is the set of {\em vertices}, and $\E(\H)\subseteq 2^{V(\H)}$ is the set of {\em hyperedges} (or simply {\em edges}).
The size of $\H$, denoted by $|\H|$, is the sum of the sizes of $V(\H)$ and $\E(\H)$, and the \emph{dimension} of $\H$ is the largest size of an edge in $\E(\H)$.
A~{\em transversal} in a hypergraph $\H$ is a set of vertices that intersects every edge of $\H$.
An {\em independent set} in a hypergraph $\H$ is a set of vertices that contains no edge of $\H$ as a subset.
We say that a transversal (resp.~independent set) is {\em minimal} (resp.~{\em maximal}) if it is inclusion-wise minimal (resp.~maximal).
The sets of all minimal transversals and all maximal independent sets of $\H$ are respectively denoted by $Tr(\H)$ and $MIS(\H)$.
The two problems of enumerating $Tr(\H)$ and $MIS(\H)$ given $\H$ are denoted by \textsc{Trans-Enum} and \textsc{MIS-Enum}.
It is well known that a set $T\subseteq V(\H)$ is a minimal transversal of $\H$ if and only if its complementary $V(\H)\setminus T$ is a maximal independent set of $\H$, hence that both \textsc{Trans-Enum} and \textsc{MIS-Enum} are polynomially equivalent~\cite{eiter2008computational}.
In the decision version of \textsc{Trans-Enum}, one must decide, given two hypergraphs $\H$ and $\G$, whether $\H$ and $\G$ are such that $\E(\G)=Tr(\H)$.
This problem goes by the name of \textsc{Hypergraph Dualization} in the literature.
It is known that it admits a polynomial-time algorithm if and only if \textsc{Trans-Enum} or \textsc{MIS-Enum} admits an output-polynomial time algorithm~\cite{bioch1995complexity,eiter2008computational}.
For the first problem, the best known algorithm runs in quasi-polynomial time $N^{o(\log N)}$ where $N=|\H|+|\G|$.
It is due to Fredman and Khachiyan~\cite{fredman1996complexity}, and it is a long-standing open problem whether it can be improved to run in polynomial time.
It has later been improved by Tamaki in~\cite{tamaki2000space} to perform using polynomial space.
If moreover the dimension of the hypergraph is known to be bounded by some fixed constant,
then the problem is known to be solvable in polynomial time~\cite{eiter1995identifying}.
In~\cite{khardon1995translating} it is shown that \textsc{CMI} is harder than \textsc{Hypergraph Dualization}.
The goal of this paper is to prove a stronger result in acyclic convex geometries, and to show that the two problems become equivalent in a subclass of acyclic convex geometries.
We refer to~\cite{eiter1995identifying} and to the survey~\cite{eiter2008computational} for further details on hypergraph dualization.

\section{Dualization in distributive lattices}\label{sec:hardness}

We show that translating between acyclic implicational bases and meet-irreducible elements is at least as hard as the {\em dualization in distributive lattices}, a generalization of the hypergraph dualization problem.
For this problem, the existence of a quasi-polynomial time algorithm is an open question \cite{babin2017dualization,defrain2019dualization}.

Let us first define the {\em dualization problem} in lattices given by implicational bases.
Let $\L_\Sigma=(\C_\Sigma,\subseteq)$ be a lattice given by an implicational base $(X,\Sigma)$, and $\B^+$ and $\B^-$ be two antichains of $\L_\Sigma$.
We say that $\B^+$ and $\B^-$ are {\em dual} in $\L_\Sigma$ if 
\begin{equation}
    \downarrow\B^+\cap \uparrow\B^-=\emptyset\ \text{ and}\ \downarrow\B^+ \cup \uparrow \B^-=\,\C_\Sigma.\label{eq:duality}
\end{equation}
In other words, $\B^+$ and $\B^-$ are dual if $\B^+=\Max_\subseteq\{F\in \C_\Sigma \mid A\not\subseteq F\ \text{for any}\allowbreak\ A\in \B^-\}$, or equivalently if $\B^-=\Min_\subseteq\{F\in \C_\Sigma \mid F\not\subseteq A\ \text{for any}\ A\in \B^+\}$.
Then, the dualization problem in lattices given by implicational bases is defined as follows.

\begin{decproblem}
  \problemtitle{Dualization in lattices given by implicational bases (\textsc{Dual})}
  \probleminput{An implicational base $\Sig$ and two antichains $\B^+,\B^-$ of $\L_\Sigma$.}
  \problemquestion{Are $\B^+$ and $\B^-$ dual in $\L_\Sigma$?}
\end{decproblem}

Note that the lattice $\L_\Sigma$ is not given.
Only $\Sig$ is, which is a crucial point.
It is well known that \textsc{Dual} generalizes \textsc{Hypergraph Dualization} where the implicational base is empty, hence when the lattice is Boolean~\cite{eiter2008computational,nourine2015beyond}.
Recently, it was proved that the problem is co{\sf NP}-complete in general~\cite{babin2017dualization}.
This result holds even when the premises in the implicational base are of size at most two~\cite{defrain2019dualization}.
The case where the implicational base only has premises of size one captures distributive lattices.
In that case, the best known algorithm runs in sub-exponential time~\cite{babin2017dualization}, while the existence of one running in quasi-polynomial time is an open question~\cite{babin2017dualization,defrain2019dualization}.
Quasi-polynomial time algorithms are known for proper subclasses, including distributive lattices coded by product of chains~\cite{elbassioni2009algorithms}, or those coded by the ideals of an interval order \cite{defrain2019dualization}.

The next theorem suggests that translating between implicational bases and meet-irreducible elements is a tough problem, even when restricted to acyclic convex geometries.

\begin{theorem}\label{thm:dual-hard}
    There is a polynomial-time algorithm solving \textsc{Dual} in distributive lattices if there is one solving \textsc{CMI} in acyclic convex geometries.
\end{theorem}

\begin{proof}
    Let $\I_1=(X,\Sigma,\B^+,\B^-)$ be an instance of \textsc{Dual} where the lattice $\L_\Sigma=(\C_\Sigma,\subseteq)$ is assumed to be distributive.
    Without loss of generality, $\Sigma$ can be considered acyclic in that case; see~\cite{babin2017dualization,defrain2019dualization}.
    We construct an instance $\I_2=(X\cup \{z\},\Omega,\M)$ of \textsc{CMI} as follows: $\Omega$ is the implicational base on ground set $X\cup \{z\}$ constructed from $\Sigma$ by adding an implication $A\rightarrow z$ for every $A\in \B^-$.
    Clearly the obtained implicational base is acyclic.
    Then, we put 
    \begin{equation}
        \M=\{M \cup \{z\} \mid M\in \M(\L_\Sigma)\} \cup \B^+.\label{eq:dual-hard-M}
    \end{equation}
    Observe that the left and right sides of the union in Equality~\eqref{eq:dual-hard-M} are disjoint.
    A representation of the lattice $\L_{\Omega}=(\C_{\Omega},\subseteq)$ is given in Figure~\ref{fig:disthard}.
    We shall show that $\I_1$ is a positive instance of \textsc{Dual} if and only if $\I_2$ is one of \textsc{CMI}, and that it can be decided in polynomial time in the size of $\I_1$ given a polynomial-time algorithm for \textsc{CMI}.
    
    \begin{figure}
        \centering
        \includegraphics[scale=1.05]{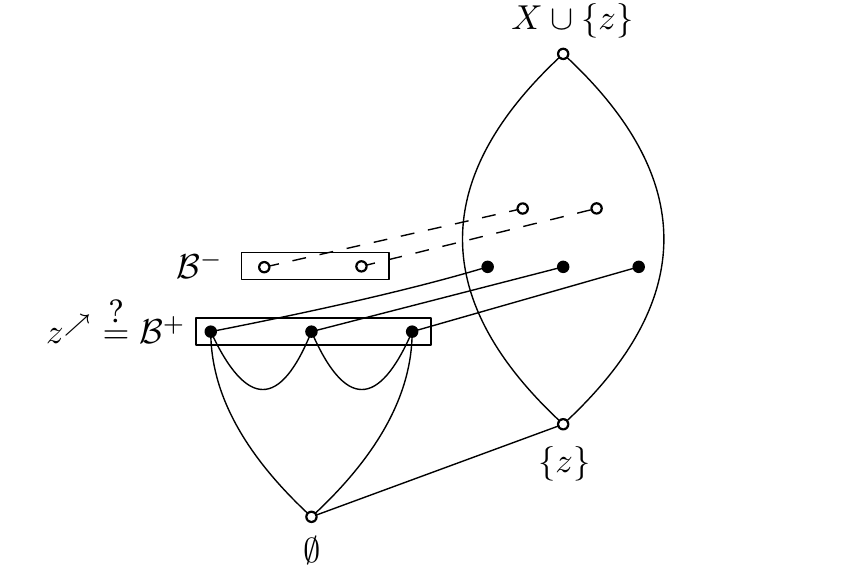}
        \caption{The construction of Theorem~\ref{thm:dual-hard}.}
        \label{fig:disthard}
    \end{figure}

    Consider the elements of $\L_{\Omega}[\uparrow\!\{z\}]$, i.e., the right part of Figure~\ref{fig:disthard}.
    Observe that $\L_\Sigma$ and $\L_{\Omega}[\uparrow\!\{z\}]$ are isomorphic.
    Consequently, there is a bijection from $\M(\L_\Sigma)$ to $\M(\L_{\Omega}[\uparrow\!\{z\}])$ given by $f: M\mapsto M\cup \{z\}$ and $f^{-1}: M'\mapsto M'\setminus \{z\}$.
    Hence
    \begin{equation}
        \{M \cup \{z\} \mid M\in \M(\L_\Sigma)\}=\M(\L_{\Omega}[\uparrow\!\{z\}]).\label{eq:dual-hard-distributive}
    \end{equation}
    Consider now the elements of $\L_{\Omega}-\!\uparrow\!\{z\}$, i.e., the left part of Figure~\ref{fig:disthard}.
    Observe that the only way for an element $N$ in such a part to be meet-irreducible is to belong to $z^\nearrow\!$, as otherwise $N$ has at least one successor in $\L_{\Omega}-\!\uparrow \{z\}$, and another one in $\L_{\Omega}[\uparrow\!\{z\}]$.
    Hence $\M(\L_{\Omega})\setminus \uparrow\!\{z\}=z^\nearrow$.
    Consequently by Equalities~\eqref{eq:dual-hard-M} and~\eqref{eq:dual-hard-distributive}, $\M=\M(\L_{\Omega})$ if and only if $\B^+=z^\nearrow$.
    As by construction, 
    \begin{equation}
        z^\nearrow=\Max_\subseteq\{F\in \C_\Sigma \mid A\not\subseteq F\ \text{for any}\ A\in \B^-\},\label{eq:dual-hard-antichain}
    \end{equation}
    we conclude that $\B^+=z^\nearrow$ if and only if $\B^+$ and $\B^-$ are dual in $\L_\Sigma$.
    Hence $\I_1$ is a positive instance of \textsc{Dual} if and only if $\I_2$ is one of \textsc{CMI}.
    Concerning the observation that $\I_1$ can be decided in polynomial time in $|\I_1|$ using a polynomial-time algorithm for \textsc{CMI}, it is a consequence of the fact that $\L_\Sigma$ being distributive, the size of $\M(\L_{\Sigma})$ is bounded by $|X|$, hence that $|\I_2|=|X\cup \{z\}|+|\Omega|+|\M|$ is bounded by a polynomial in $|\I_1|$.
    This concludes the proof.
\end{proof}

As a consequence, there is no output quasi-polynomial time algorithm for \textsc{CCM}, nor \textsc{SID}, unless there is one solving the dualization in distributive lattices in quasi-polynomial time, a long-standing open problem~\cite{babin2017dualization,defrain2019dualization}.

As for the computation of $j^\nearrow$ given an acyclic implicational base $\Sig$ and some element $j\in X$, a corollary of Equality~\eqref{eq:dual-hard-antichain} in the proof of Theorem~\ref{thm:dual-hard} is that the task is even harder than the dualization in lattices given by acyclic implicational bases (a proper superclass of distributive lattices), whenever $\Sigma$ is acyclic.
To~the best of our knowledge, no better algorithms than exponential naive ones are known on such class.
This suggests that the technique employed in~\cite{beaudou2017algorithms} and pointed by Proposition~\ref{prop:meets-reduces-to-colors} for the enumeration of meet-irreducible elements from join-irreducible elements may not be efficient in acyclic convex geometries.
However, we will show in Section~\ref{sec:mis} that it can work in a subclass of acyclic convex geometries.

\section{Ranked convex geometries}\label{sec:ranked}

Let $\Sig$ be an implicational base.
A {\em rank function} on $\Sig$ is a function $\rho:X \rightarrow \mathbb{N}$ such that if $A\rightarrow b\in \Sigma$ and $a\in A$, then $\rho(a)=\rho(b)+1$.
We say that $\Sig$ is {\em ranked} if it admits a rank function.
It~is not hard to see that ranked implicational bases are acyclic, hence that they define a subclass of acyclic convex geometries.
We say that a convex geometry is {\em ranked} if it admits a ranked implicational base. 
Observe that the construction of Theorem~\ref{thm:dual-hard} is not ranked, as not all distributive lattices admit a ranked implicational base, and in addition the premises of the constructed implications may not contain elements of a same rank.

In this section, we prove structural properties on acyclic and ranked convex geometries.
Let us first give a result from~\cite{hammer1995quasi} (see also~\cite{wild2017joy}) which plays a central role in the remaining of the section.
We recall that a set $A\subseteq X$ is a minimal generator of $b$ if it is minimal satisfying $b \in \phi(A)$.

\begin{proposition}[\cite{hammer1995quasi}]\label{prop:hammer}
    Let $(X,\phi)$ be an acyclic convex geometry.
    Then $(X,\phi)$ admits a unique irredundant implicational base of minimal generators which is unit-minimum.
    Furthermore, it can be computed in quadratic time from any unit implicational base.
\end{proposition}

In the following, we call {\em critical base} the irredundant implicational base of minimal generators of an acyclic convex geometry given by Proposition~\ref{prop:hammer}.
Accordingly, we call {\em critical}\footnote
{%
    The terminology coincides with the notion of {\em critical} from~\cite{dietrich1987circuit} (the second part of Proposition~\ref{prop:acyclic-gen}) and {\em essential} from~\cite{hammer1995quasi} (minimal generators that belong to every implicational bases of minimal generators) in acyclic convex geometries.
    The critical base is also known to be \emph{optimum} after implications with the same premises are combined into a single implication; see~\cite{hammer1995quasi,adaricheva2013ordered,adaricheva2017optimum,wild2017joy} for more details on these notions.
} 
a minimal generator that belongs to this base, and {\em redundant} one that does not.
The next two propositions can be inferred from other results in matroid and antimatroid theory~\cite{dietrich1987circuit,korte2012greedoids,nakamura2013prime}.
We nevertheless reprove them here for self-containment.

\begin{proposition}\label{prop:acyclic-gen}
    Let $(X,\phi)$ be an acyclic convex geometry and $A\subseteq X$ be a minimal generator of $b\in X$. 
    Then $A$ is redundant if and only if its closure $\phi(A)$ contains another minimal generator of $b$.
\end{proposition}

\begin{proof}
    Let $A$ be a redundant minimal generator of $b$ and $\Sigma$ be the critical base of $(X,\phi)$ given by Proposition~\ref{prop:hammer}.
    Since $A \rightarrow b$ does not belong to $\Sigma$ and $b \in \phi(A)$, there exists an implication $C \rightarrow b \in \Sigma$, $A\neq C$ such that $C \subseteq \phi(A)$.
    
    Let $A$ be a minimal generator of $b$ and suppose that there exists another minimal generator $C$ of $b$ such that $C \subseteq \phi(A)$.
    Consider the critical base $\Sigma$ of $(X,\phi)$ given by Proposition~\ref{prop:hammer}.
    Since $A$ and $C$ are minimal generators, $C\setminus A$ is not empty.
    Since $C\subseteq \phi(A)$, for each $c \in C\setminus A$ there are implications in $\Sigma$ leading to $c$ from~$A$.
    Since $C$ is a minimal generator of $b$, there are implications in $\Sigma$ leading to $b$ from~$C$.
    Consequently, there exists a sequence of implications in $\Sigma$ that does not contain $A\rightarrow b$ and that nevertheless produces $b$ from $A$.
    We conclude that $A\rightarrow b$, if in $\Sigma$, is redundant, a contradiction.
    Hence $A$ is redundant.
\end{proof}

\begin{remark}\label{rem:acyclic-gen}
    By the choice of $\Sigma$ in the proof of Proposition~\ref{prop:acyclic-gen}, the minimal generator of~$b$ in Proposition~\ref{prop:acyclic-gen} can be required to be critical.
\end{remark}

A corollary of this proposition is the next characterization.

\begin{corollary}\label{cor:acyclic-gen}
    Let $(X,\phi)$ be an acyclic convex geometry and $A\subseteq X$ be a minimal generator of $b\in X$. 
    Then $A$ is redundant if and only if there exists $a \in A$ such that $\phi(A) \setminus \{a, b\}$ implies $b$.
\end{corollary}

We deduce the following link between critical minimal generators and ranked implicational bases.

\begin{proposition}\label{prop:acyclic-superset}
    Let $(X,\phi)$ be an acyclic convex geometry, $A\subseteq X$ be a critical minimal generator of $b\in X$, and $\Sigma$ be any implicational base of $(X,\phi)$.
    Then there exists $C \rightarrow b \in \Sigma$ such that $A \subseteq C$.
\end{proposition}

\begin{proof}
    Let $A$ be a critical minimal generator of $b$ and assume for contradiction that there is no implication $C \rightarrow b \in \Sigma$ such that $A \subseteq C$.
    In particular since $b\not\in A$ and $b\in \phi(A)$ there exists at least one implication $D \rightarrow b$ in $\Sigma$ such that $D\subseteq \phi(A)$.
    By hypothesis $A\not\subseteq D$.
    Hence by acyclicity of $(X,\phi)$ we deduce $D\subset \phi(A)$.
    Then there exists $D' \subseteq D \subset \phi(A)$ such that $D'$ is a minimal generator of $b$.
    By Proposition~\ref{prop:acyclic-gen}, it contradicts $A$ being critical.
\end{proof}

We are now ready to state the main result of this section, which is of interest as far as the computation of a unit-minimum ranked implicational base is concerned (SID).

\begin{theorem}\label{thm:irr-ranked}
    Let $(X, \phi)$ be an acyclic convex geometry.
    Then $(X, \phi)$ is ranked if and only if its critical base is ranked.
\end{theorem}

\begin{proof}
    The if part follows from the definition.
    We prove the other direction.
    Let us assume that $(X, \phi)$ is ranked and consider its critical base $\Sigma^*$ given by Proposition~\ref{prop:hammer}.
    Consider now any ranked implicational base $\Sigma$ of $(X, \phi)$.
    Then by Proposition~\ref{prop:acyclic-superset} for every implication $A \rightarrow b \in \Sigma^*$ there exists $A' \rightarrow b\in \Sigma$ such that $A\subseteq A'$.
    Hence if $\Sigma^*$ is not ranked, it must be that the elements in such $A$'s together with their conclusion $b$ do not admit a rank function in $\Sigma^*$.
    As all such elements appear as is in $\Sigma$, they cannot admit a rank function in $\Sigma$ neither, contradicting $\Sigma$ being ranked.
    Consequently $\Sigma^*$ must be ranked.
\end{proof}

We now consider the problem of deciding whether an acyclic convex geometry is ranked, from an implicational base and from its meet-irreducible elements.
We show that the first problem lies in {\sf P} while the second is in co{\sf NP}.
Whether the second problem is in {\sf P} or is co{\sf NP}-hard is left as an open problem.
As a preliminary remark, observe that the implicational base obtained by disjoint union of ranked implicational bases is ranked.
A~corollary is that every closure system obtained by product of chains is ranked, a class of interest in~\cite{elbassioni2009algorithms}.

\begin{proposition}\label{prop:test}
    Checking whether an implicational base $(X,\Sigma)$ admits a rank function~$\rho$, and computing $\rho$ if it does, can be done in polynomial time in the size of $(X,\Sigma)$.
\end{proposition}

\begin{proof}
    Let $G(\Sigma)$ be the implication-graph of $\Sigma$ that can be computed in polynomial time in the size of $(X,\Sigma)$.
    Observe that we can restrict ourselves to the case where $G(\Sigma)$ is connected, as otherwise we handle each connected component independently.
    The algorithm proceeds as follows.
    At first it picks a vertex $x$ of $G(\Sigma)$ and set $\rho(x)=n$ ($n$ is chosen to ensure that we cannot attribute a negative rank to a vertex in the next steps, hence that $\rho: X\rightarrow \mathbb{N}$).
    Then, for every unmarked vertex $x$ with a rank, the algorithm marks $x$ and extend $\rho$ to every $y\in x^-\cup x^+$, until no such vertex exists.
    Note that extending $\rho$ is deterministic by definition.
    Hence if a conflict is detected during the procedure, then we can certificate that $\Sigma$ is not ranked.
    Otherwise, a rank function is computed.
\end{proof}

\begin{proposition}\label{prop:conp}
    Deciding whether an acyclic convex geometry is ranked from its meet-irreducible elements belongs in co{\sf NP}.
\end{proposition}

\begin{proof}
    We describe a polynomial-size certificate that can be checked in polynomial time in order to answer negatively.
    Such a certificate is a set $\{{A_1 \rightarrow b_1},\allowbreak{}\, \dots\,,\, {A_k \rightarrow b_k}\}$ of critical implications that do not admit a rank function.
    Recall that by Theorem~\ref{thm:irr-ranked} the convex geometry is not ranked if and only if these implications exist.
    Furthermore, they must induce an undirected cycle in $G(\Sigma)$, and only one such cycle is needed to answer negatively.
    Hence $k$ is bounded by $|X|$.
    One has to check first that each of these implications is indeed a critical minimal generator, following the next three steps.
    To decide whether $A_i \rightarrow b_i$ is valid, we compute the closure of $A_i$ using meet-irreducible elements, as described in the preliminaries.
    To check that $A_i$ is a minimal generator of $b_i$, we test whether $b_i\not\in \phi(A_i\setminus \{a\})$ for all $a \in A_i$.
    As for the critical property, we check according to Corollary~\ref{cor:acyclic-gen} whether there exists $a \in A_i$ such that $\phi(A_i) \setminus \{a, b_i\}$ implies $b_i$.
    Then we launch the polynomial-time procedure given in Proposition~\ref{prop:test} to check whether these implications admit a rank function.
\end{proof}

\section{Enumerating the meet-irreducible elements in ranked ACG}\label{sec:mis}

We give an output quasi-polynomial time algorithm, based on hypergraph dualization, for the enumeration of meet-irreducible elements (CCM) in closure systems given by a ranked implicational base.
If in addition the implicational base is known to have its dimension bounded by some fixed constant, then the algorithm performs in output-polynomial time.

Let $\C_{\Sigma}$ be a closure system given by a ranked implicational base $(X,\Sigma)$ with rank function $\rho$ and closure operator~$\phi$.
Note that by Proposition~\ref{prop:hammer} and Theorem~\ref{thm:irr-ranked} we can assume $\Sigma$ to be the critical base of $(X,\phi)$.
In the following and for every $B \subseteq X$ we put
\[
    B^{\nn}\! = \Max_\subseteq \{ C \in \C_{\Sigma} \mid C\cap B=\emptyset\}.
\]
Observe that this definition generalizes the $\nearrow$ relation for singletons, as the following remark states.

\begin{remark}\label{rem:generalization}
    Let $B\subseteq X$. Then $B^{\nn}=B^\nearrow$ if $|B|=1$.
    Hence the elements in $B^{\nn}$ are meet-irreducible in that case.
\end{remark}

In the following, we say that a set $B\subseteq X$ is {\em ranked} with respect to $\rho$ if every element in $B$ shares the same rank according to $\rho$, that is, if $\rho(a)=\rho(b)$ for all $a,b\in B$.
Then to every ranked set $B$ we associate $\rho(B)$ the rank of one of its elements (or an arbitrary rank if $B$ is empty), and $\H_B$ the hypergraph defined on vertex set $V(\H_B)=\{x\in X \mid \rho(x)=\rho(B)+1\}$ and edge set
\[
    \E(\H_B)=\{A \mid A\rightarrow b\in \Sigma\ \text{for some}\ b\in B\}.
\]
Note that every vertex of $\H_B$ has the same rank.
In general, $V(\H_B)$ may be a proper superset of $\bigcup \E(\H_B)$: 
the elements that do not belong to any hyperedge of $\H_B$ are of interest in the following.
We will show that the maximal independent sets of $\H_B$ allow to define a partition of $B^{\nn}\!$.
For any $C\subseteq X$ we denote by $C_i$ the elements of $C$ with rank $i$.

\needspace{0.5in}
\begin{proposition}\label{prop:maximum-rank}
  Let $B\subseteq X$ be a ranked set of maximum rank in $\Sigma$.
  Then $B^{\nn}=\{C\}$ where $C=\{x\in X \mid \rho(x)\leq \rho(B),\ x\not\in B\}$.
\end{proposition}

\begin{proof}
  As $B$ is of maximal rank, no implication in $\Sigma$ has $b \in B$ for conclusion.
  Since $\Sigma$ is acyclic, no element $x\in X$ such that $\rho(x)\leq \rho(B)$ with $x\not\in B$ can imply $b\in B$.
  The proposition follows.
\end{proof}

In the following, let us put $k=\Max\{\rho(x) \mid x\in X\}$ to be the maximum rank of $\Sigma$.
The aforementioned partition is the following.

\begin{theorem}\label{thm:partition} 
  Let ${B\subseteq X}$ be a ranked set of rank ${i<k}$, and $\H_B$ be its associated hypergraph.
  Then there is a partition of $B^{\nn}$ given by
  \[
    \big\{\{ C\in B^{\nn} \colon C_{i+1}=S\} \mid S \in MIS(\H_B)\big\}.
  \]
\end{theorem}

The proof of Theorem~\ref{thm:partition} is decomposed into Lemmas~\ref{lem:partition-lem1} and~\ref{lem:partition-lem2}. 
The partition is illustrated in Figure~\ref{fig:partition}.

\needspace{0.2in}
\begin{lemma} \label{lem:partition-lem1} 
    Let $B\subseteq X$ be a ranked set of rank $i<k$, and $\H_B$ be its associated hypergraph. 
    Then to every $C\in B^{\nn}$ corresponds $S\in MIS(\H_B)$ such that $C_{i+1}=S$.
\end{lemma}

\begin{proof} 
    Let $C\in B^{\nn}$ and $S=C_{i+1}$.
    Observe that since $C$ is closed and as it does not intersect~$B$, $S$ is an independent set of $\H_B$.
    We show that it is maximal.
    Let $x \in V(\H_B)$ such that $x\not\in S$.
    Then $x\not\in C$ and by maximality of $C$, $C\cup \{x\}$ implies an element of~$B$.
    Since $\rho(x)=i+1$ and as $\Sigma$ is ranked, every implication $A \rightarrow y \in \Sigma$ with $x$ in its premise has $A$ of rank $\rho(S)=i+1$ and $y$ of rank $\rho(B)=i$.
    As $\rho(y)=\rho(B)$ no such $y$ belongs to the premise of an implication having $b\in B$ for conclusion.
    Hence there must be $A \rightarrow b \in \Sigma$ such that $x\in A$, $b\in B$ and consequently such that both $A \in \E(\H_B)$ and $A\subseteq S\cup \{x\}$.
    Therefore $S\cup \{x\}$ is no longer an independent set of $\H_B$. 
\end{proof}

\begin{lemma} \label{lem:partition-lem2} 
    Let $B\subseteq X$ be a ranked set of rank $i<k$, and $\H_B$ be its associated hypergraph. 
    Then to every $S\in MIS(\H_B)$ corresponds some $C\in B^{\nn}$ such that $C_{i+1}=S$.
\end{lemma}

\begin{proof} 
    Consider $S \in MIS(\H_B)$.
    Recall that every implication having an element $b\in B$ for conclusion has its premise in $\E(\H_B)$. 
    Furthermore since $\Sigma$ is ranked, the closure of $S$ does not contain any other elements than those of $S$ at rank $i+1$.
    Consequently as $S$ is an independent set of $\H_B$ it does not imply any element of $B$.
    Hence there exists $C \in B^{\nn}$ such that $S\subseteq C$.
    Now since $S$ is maximal, adding any $x\in V(\H_B)\setminus S$ to $S$ results into implying some $b\in B$.
    We conclude that $C_{i+1}=S$.
\end{proof}

\begin{figure}
    \center
    \includegraphics[scale=1.05]{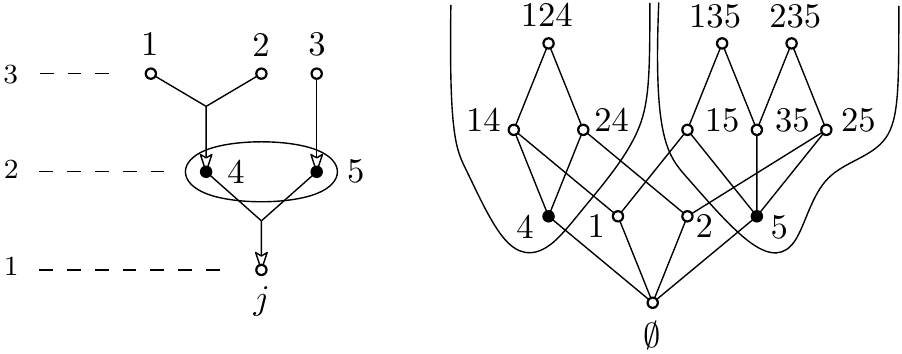}
    \caption{The situation of Theorem \ref{thm:partition}. 
    On the left a ranked implicational base $(X, \Sigma)$ and its rank function (dashed). 
    On the right the set of closed sets of $\Sigma$ not containing~$j$, ordered by inclusion. 
    Maximal independent sets of $\E(\H_{j}) = \{\{4, 5\}\}$ are $\{4\}$ and $\{5\}$.
    They partition $\smash{j^{\nn}} = j^{\nearrow} = \{\{1,2,4\}\}\uplus\{\{1,3,5\}, \{2,3,5\}\}$.}\label{fig:partition}
\end{figure}

The next lemma offers a recursive characterization of $B^{\nn}$ based on the partition defined in Theorem~\ref{thm:partition}.
Recall that since $\Sigma$ is acyclic, no element $x\in X$ such that $\rho(x)\leq \rho(B)$, $x\not\in B$ takes part in a minimal generator of an element of $B$. 
In particular, every such $x$ belongs to all $C\in B^{\nn}$.
For any ranked set $S$ we put $\hat{S}$ to be the elements of rank $\rho(S)$ not in $S$, i.e., $\hat{S}=\{x\in X \mid \rho(x)=\rho(S),\ x\not\in S\}$.

\begin{lemma} \label{lem:characterization}
    Let $B \subseteq X$ be a ranked set of rank $i<k$, and $\H_B$ be its associated hypergraph. 
    Let $S\in MIS(\H_B)$.
    Then
    \begin{equation}
        \{ C\in B^{\nn}\!\mid C_{i+1}=S\} = \{I \setminus B \mid I \in \hat{S}\,\!^{\nn}\}.\label{eq:characterization}
    \end{equation}
\end{lemma}

\begin{proof}
    We show the first inclusion.
    Let $C \in B^{\nn}$ such that $C_{i+1}=S$. 
    Then $\hat{S}\cap C=\emptyset$.
    Since $\Sigma$ is acyclic, $C$ contains every $x \in X$ such that $\rho(x)\leq i$ and $x\not\in B$.
    Consequently, $I=C\cup B$ is closed.
    Since $\rho(B)=i$ and $\rho(\hat{S})=i+1$, $I$ does not imply any element of~$\hat{S}$.
    We show that $I$ is maximal with this property.
    Let $x\not\in I$.
    Because $C\in B^{\nn}$, there exists $b\in B$ such that $C \cup \{x\}$ implies $b$.
    Hence $\phi(C\cup \{x\})$ must contain an edge of $\H_B$ (i.e.,~the premise of an implication $A\rightarrow b$, $b\in B$) and so does $\phi(I\cup \{x\})$.
    Since $C_{i+1}=S$ and as by hypothesis $S$ is a maximal independent set of $\H_B$, either $x\in \hat{S}$ or $I \cup \{x\}$ implies $s$ for some $s \in \hat{S}$, proving the first inclusion.
    
    We show the other inclusion.
    Let $I \in \hat{S}\,\!^{\nn}$.
    Then $\hat{S}\cap I=\emptyset$.
    Since $\Sigma$ is ranked, $I$ contains every $x \in X$ such that $\rho(x) \leq i+1$ and $x\notin \hat{S}$.
    In particular $B,S\subseteq I$.
    Since $S$ is an independent set of $\H_B$, $C=I \setminus B$ does not imply $b$, for any $b \in B$.
    Hence it is closed and $C_{i+1}=S$.
    We show that it is maximal with this property.
    Let $x\not\in C$.
    Either $x\in I$ (and then $x\in B$) or $x \notin I$.
    Obviously, $x$ cannot be added to $C$ in the first case without intersecting $B$.
    Let us assume that $x\not\in I$.
    Then $\rho(x) \geq \rho(\hat{S})$.
    Since $I \in \hat{S}\,\!^{\nn}$ and by maximality of $I$, either $x \in \hat{S}$ or there exists $s \in \hat{S}$ such that $I \cup \{x\}$ implies $s$, hence such that $C\cup \{x\}$ implies~$s$.
    As by assumption $S$ is a maximal independent set of $\H_B$, $C\cup \{x\}$ implies some $b\in B$ in both cases, concluding the proof.
\end{proof}

In what follows given $C\subseteq X$ we put $C_{>i}=\{x\in C\mid \rho(x)>i\}$.
Since the elements on the left and right parts of Equation~\eqref{eq:characterization} only differ on $B$, and as $\rho(B)=i$, a corollary of Lemma~\ref{lem:characterization} is the following.
It is of interest as far as the proof of our algorithm is concerned.

\begin{corollary} \label{cor:characterization}
    Let ${B \subseteq X}$ be a ranked set of rank ${i<k}$, and $\H_B$ be its associated hypergraph. 
    Let $S\in MIS(\H_B)$.
    Then
    \[
        \{ C_{>i} \mid C\in B^{\nn}\!\text{\emph{ s.t.}}~C_{i+1}=S\} = \{ I_{>i} \mid I\in \hat{S}\,\!^{\nn}\}.
    \]
\end{corollary}

We now describe an algorithm which enumerates the set $\M(\C_\Sigma)$ of meet-irre\-ducible elements of $\C_\Sigma$ given $(X, \Sigma)$ whenever it is ranked.
The algorithm proceeds as follows.
For every $j\in X$, it computes $j^{\nn}=j^{\nearrow}$ recursively rank by rank relying on the partition and the characterizations of Theorem~\ref{thm:partition}, Lemma~\ref{lem:characterization} and Corollary~\ref{cor:characterization}.
Recall that since $\C_{\Sigma}$ is a convex geometry, $\{j^{\nearrow} \mid j\in X\}$ is a partition of $\M(\C_\Sigma)$, hence that every solution is obtained by such a procedure, without duplication.
The computation of $j^{\nn}$ is described in Algorithm~\ref{algo:main}.
The correctness of the whole procedure is proved in Theorem~\ref{thm:DFS-mis}.

\begin{algorithm2e}
  \SetAlgoLined
  
  $k\hspace{0.09cm}\leftarrow \Max\{\rho(x) \mid x\in X\}$\;
  
  $C\leftarrow \{x\in X \mid \rho(x)\leq \rho(B),\ x\not\in B\}$\;\label{line:preprocessing}
    
  \texttt{RecENUM}($\Sigma,B,C$)\;
  
  \vspace{0.2cm}

  \SetKwProg{myproc}{Procedure}{}{}
  \myproc{{\em \texttt{RecENUM}($\Sigma,B,C$)}\label{line:receive}}{
        {\bf if} $\rho(B)=k$ {\em\bf then output} $C$ {\em\bf and return}\;\label{line:output}
        
        \For{{\em\bf all} $S\in MIS(\H_B)$\label{line:forall}}
        {
            
            $\hat{S}\leftarrow \{x\in X \mid \rho(x)=\rho(S),\ x\not\in S\}$\;\label{line:chapeau}
            
            \texttt{RecENUM}($\Sigma, \hat{S}, C\cup S$)\;\label{line:rec}
        }
  }

  \caption{A recursive algorithm enumerating the set $B^{\nn}$ given a ranked implicational base $\Sig$ and a ranked set $B\subseteq X$.}\label{algo:main}
\end{algorithm2e}

\begin{theorem}\label{thm:DFS-mis}
    Let $f,g:\mathbb{N}\rightarrow\mathbb{N}$ be two functions and $\C_\Sigma$ be the closure system of a given ranked implicational base $\Sig$.
    Then there is an algorithm enumerating the set $\M(\C_\Sigma)$ of meet-irreducible elements of $\C_\Sigma$ with delay
    \[
        O(|X| \cdot f(|X|+|\Sigma|)),
    \]
    $O(|X|\cdot g(|X|+|\Sigma|))$ space, and after $O(|X|\cdot|\Sigma|)$ preprocessing time whenever there is one enumerating $MIS(\H)$ with delay $f(|\H|)$ and space $g(|\H|)$ given \hbox{a hypergraph $\H$.}
\end{theorem}

\begin{proof} 
    Recall that by Proposition~\ref{prop:meets-reduces-to-colors}, there is an algorithm enumerating $\M(\C_\Sigma)$ with delay at most $2\cdot f(|X|+|\Sigma|)$ whenever there is one enumerating $j^\nearrow$ with delay
    $f(|X|+|\Sigma|)$ for any $j\in X$.
    Hence, it is sufficient to show that there is an algorithm enumerating $j^{\nn}=j^\nearrow$ with $O(|X| \cdot f(|X|+|\Sigma|))$ delay, $O(|X| \cdot g(|X|+|\Sigma|))$ space and after $O(|X|\cdot |\Sigma|)$ preprocessing time to prove the theorem.
    We shall show that Algorithm~\ref{algo:main} correctly outputs this set and that it performs within these time and space bounds. 
    
    We show correctness.
    Assume without loss of generality that $\Min\{\rho(x) \mid x\in X\}=0$ and let us put $k=\Max\{\rho(x) \mid x\in X\}$.
    Let $T[B,C]$ denote the recursive execution-tree of \texttt{RecENUM}$(\Sigma,B,C)$ for $\Sigma$, a ranked set $B$ and $C \subseteq X$.
    We show by induction on the rank $i$ of $B$ and for every $C\subseteq X$ that the set $\{C\cup C_{>i}^* \mid C^*\in B^{\nn}\}$ is output at the leaves of $T[B,C]$.
    Note that we aim here to prove both implications at the same time.
    Let us first consider the case $i=k$.
    By Proposition~\ref{prop:maximum-rank}, $B^{\nn}=\{C^*\}$ where $C^*=\{x\in X \mid \rho(x)\leq \rho(B),\ x\not\in B\}$ and $C_{>i}^*=\emptyset$ in that case.
    Also $C$ is output Line~\ref{line:output} of the algorithm.
    Hence the property holds for $i=k$, the tree $T[B,C]$ being reduced to a leaf.
    Let us assume that the property holds for every rank greater than~$i$.
    We show that it holds for $i$.
    Consider $B$ of rank~$i$.
    Recall that by Theorem~\ref{thm:partition}, $\{\{C^*\in B^{\nn} \colon C^*\!\!_{i+1}=S\} \mid S \in MIS(\H_B)\}$ is a partition of $B^{\nn}$, and furthermore by Corollary~\ref{cor:characterization}, 
    \begin{align*}
        \{ C_{>i}^* \mid C^*\in B^{\nn} \!\text{ s.t.}~C_{i+1}^*=S\} &= \{\,\,\, I_{>i}\hspace{0.168cm} \mid I\in \hat{S}\,\!^{\nn}\}.\hspace{0.34cm}
    \end{align*}
    By inductive hypothesis since $\rho(\hat{S})=i+1$, the set $\{C\cup I_{>i+1} \mid I\in \hat{S}\,\!^{\nn}\}$ is 
    output by a call to \texttt{RecENUM}($\Sigma, \hat{S}, C$) for every $C\subseteq X$.
    Since $I_{i+1}=S$ whenever $I\in \hat{S}\,\!^{\nn}\!\!$, 
    \begin{align*}
        \{C \cup C^*\!\!_{>i} \mid C^*\in B^{\nn}\!\text{ s.t.}~C^*\!\!_{i+1}=S\}&=\{C\cup S\cup I_{>i+1}\mid I\in \hat{S}\,\!^{\nn}\}\\
        &=\{\hspace{0.536cm}C\cup I_{>i}\hspace{0.593cm}\mid I\in \hat{S}\,\!^{\nn}\}
    \end{align*}
    for every $C\subseteq X$.
    Hence by Theorem~\ref{thm:partition} the set $\{C\cup C_{>i}^* \mid C^*\in B^{\nn}\}$ is output by a call to \texttt{RecENUM}($\Sigma, \hat{S}, C\cup S$) for every $S\in MIS(\H_B)$ Lines~\ref{line:forall},~\ref{line:chapeau} and~\ref{line:rec}.
    Consequently it is obtained at the leaves of $T[B,C]$.
    This concludes the induction.
    Now since every $C^*\in B^{\nn}$ intersects the $\rho(B)$ first ranks on $C=\{x\in X \mid \rho(x)\leq \rho(B),\ x\not\in B\}$, the set $B^{\nn}$ is output by an initial call to \texttt{RecENUM}($\Sigma, B, C$).
    
    Let us now consider the complexity of the algorithm.
    Observe that a rank function of $(X,\Sigma)$ and an ordered sequence of the elements of $X$ according to their rank can be computed in $O(|X|\cdot|\Sigma|)$ preprocessing time and $O(|X|^2)$ space, so that the computations of $C$ and $\hat{S}$ Lines~\ref{line:preprocessing} and~\ref{line:chapeau} can be achieved in $O(|X|)$ time.
    As for the computation of $\H_B$, it can be achieved in $O(|X|+|\Sigma|)$ time and space by indexing implications of $\Sigma$. 
    By assumption since $|\H|\leq |X|+|\Sigma|$, every $S\in MIS(\H_B)$ can be enumerated with $f(|X|+|\Sigma|)$ delay and using $g(|X|+|\Sigma|)$ total space.
    Assuming that $f,g\in \Omega(|X|+|\Sigma|)$, the computation of $\H_B$ Line~\ref{line:forall}, $\hat{S}$ Line~\ref{line:chapeau} is meaningless in term of time and space compared to that of $S\in MIS(\H_B)$.
    This also holds for the space needed by the preprocessing steps.
    Now, since the depth of the recursive tree is bounded by $|X|$, the time and space complexities follow.
\end{proof}

A trivial bound of $f(|\H|+|MIS(\H)|)$ on the delay of an output-poly\-nomial time algorithm enumerating $MIS(\H)$ within the same time yields an output quasi-polynomial time algorithm enumerating meet-irreducible elements in closure systems given by ranked implicational bases, using the algorithm of Fredman and Khachiyan \cite{fredman1996complexity} as a subroutine for the enumeration of $MIS(\H)$.
If moreover the dimension of $\Sigma$ is known to be bounded by some fixed constant, then the described algorithm runs in output-polynomial time using the algorithm of Eiter and Gottlob \cite{eiter1995identifying}.
It performs with polynomial delay whenever $\Sigma$ is of dimension two using the algorithm of Tsukiyama~\cite{tsukiyama1977new} on the enumeration of maximal independent sets in graphs.

It is worth mentioning that the partition given by Theorem~\ref{thm:partition} holds for more general closure spaces than ranked convex geometries. 
One such convex geometry is given in Example~\ref{ex:open-pbs}.
Therefore, the question of Open problem~\ref{op:meet} naturally arises.

\begin{example}\label{ex:open-pbs}
    We consider the implicational base $\Sigma = \{ 1 \rightarrow 2, 1 \rightarrow 3,\allowbreak{} {2 \rightarrow 4},\allowbreak{} {34 \rightarrow 5} \}$. 
    Despite the fact that $\Sigma$ is not ranked, Theorem~\ref{thm:partition} holds.
    Hence, an algorithm in the fashion of Algorithm~\ref{algo:main} correctly outputs the meet-irreducible elements in that case.
\end{example}

\begin{problem}\label{op:meet}
    For which closure spaces does Theorem~\ref{thm:partition} hold?
\end{problem}

Recall however that by Theorem~\ref{thm:dual-hard} and Equality~\eqref{eq:duality} the strategy of enumerating $j^\nearrow$ for every $j\in X$ in acyclic convex geometries is harder than the dualization in lattices given by acyclic implicational bases. 

\bigskip

\section{Constructing the ranked implicational base}\label{sec:tr}

In this section, we give an output quasi-polynomial time algorithm constructing the critical base of a ranked convex geometry given by the set of its meet-irreducible elements (SID).
Recall by Theorem~\ref{thm:irr-ranked} that such an implicational base is unit-minimum.

In what follows, let $(X,\phi)$ be a ranked convex geometry, $(X,\C_\phi)$ be its closure system and $\M=\M(\C_\phi)$ be the set of its meet-irreducible elements.
Let $(X,\Sigma)$ denote the critical base we are willing to compute.
Then for every $j \in X$ we put
\[ 
    \pred(j) = \{ a \in X \mid \exists A \rightarrow j \in \Sigma, \ a \in A \}. 
\]
In other words, $\pred(j)$ is the set of elements of $X$ belonging to some critical minimal generator of $j$.
Recall that since $\Sigma$ is ranked, there is no minimal generator $A$ of $b$ such that $a\in A$ and $a,b\in \pred(j)$.
Then, to $j \in X$ we associate the hypergraph $\H_j$ defined on vertex set $V(\H_j)=X \setminus \bigcap_{M \in j^{\nearrow}} M$ and edge set
\[
    \E(\H_j) = \{ X\setminus M \mid M \in j^{\nearrow} \}.
\]
It is well known that the minimal transversals of this hypergraph are exactly the minimal generators of $j$.
See for instance~\cite{mannila1992design, wild1994theory, bertet2018lattices,habib2018representation}.
In the following we show that we can restrict this hypergraph so that its minimal transversals are exactly the critical minimal generators of $j$. 
The next lemma provides a characterization of $\pred(j)$ that depends on $\M$.

\begin{lemma}\label{lem:pred}
    Let $j \in X$, $M \in j^{\nearrow}$ and $a \notin M$.
    Then $a \in pred(j)$ if and only if $M \cup \{a, j\}$ is closed.
\end{lemma}

\begin{proof}
    We prove the first implication.
    Consider $j \in X$, $M \in j^{\nearrow}$ and $a \notin M$.
    Note that one such $M$ exists as by previous remarks $a$ belongs to at least one minimal transversal of $\H_j$.
    Let us assume that $a \in \pred(j)$.
    Since $\Sigma$ is ranked, for any $x$ such that $\rho(x) \leq \rho(j)$, $x\neq j$ we have $x \in M$.
    This holds in particular for every element $y\in X$, $y\neq j$ such that there is a minimal generator $A$ of $y$ with $a \in A$ as $\rho(y) < \rho(a)$ and $\rho(a) = \rho(j) + 1$ in that case.
    Consequently, $M \cup \{a, j\}$ is closed.
    
    We prove the other implication.
    Let $j\in X$, $M\in j^\nearrow$ and $a\notin M$ such that $M \cup \{a, j\}$ is closed.
    By definition $M$ is closed, it does not contain $j$, and it is maximal with this property.
    In particular $M \cup \{a\}$ implies $j$, that is, $j \in \phi(M \cup \{a\})$.
    By Proposition~\ref{prop:acyclic-gen} and Remark~\ref{rem:acyclic-gen}, there exists a critical minimal generator $E$ of $j$ such that $E \not\subseteq M$ and $E \subseteq \phi(M \cup\{a\}) = M \cup \{a, j\}$ as $M \cup \{a, j\}$ is closed by assumption.
    Then $a \in E$ and consequently $a \in \pred(j)$.
\end{proof}

Observe that the second part of the proof holds in the more general context of acyclic convex geometries.
We are now ready to characterize the critical base of a ranked convex geometry.
Given a hypergraph $\H$ and $S\subseteq V(\H)$ we note $\H[S]$ the {\em hypergraph induced} by $S$ defined by $V(\H[S])=S$ and $\E(\H[S])=\{E\cap S \mid E\in \H\}$.

\begin{theorem}\label{thm:mintr-characterization}
    Let $(X,\Sigma)$ be the critical base of a ranked convex geometry.
    Let $j\in X$.
    Then $A$ is a critical minimal generator of $j$ if and only if $A\in Tr(\H_j[\pred(j)])$.
\end{theorem}

\begin{proof}
    We prove the first implication.
    Since $A$ is a critical minimal generator of $j$, it is a minimal transversal of $\H_j$, hence of $\H_j[\pred(j)]$ as $A\subseteq \pred(j)$ by definition.
    
    We prove the other implication.
    Let $A$ be a minimal transversal of $\H_j[\pred(j)]$.
    Note that $E\cap \pred(j)\neq\emptyset$ for all $E\in \E(\H_j)$, as $\pred(j)$ is a transversal of $\H_j$.
    Hence $A$ is a minimal transversal of $\H_j$.
    Assume toward a contradiction that $A$ is a redundant minimal generator of $j$.
    Then by Corollary~\ref{cor:acyclic-gen} there exists $a \in A$ such that $\phi(A) \setminus \{a,j\}$ implies~$j$.
    That is, there is at least one critical minimal generator $E$ of $j$, $a \notin E$ and $E \subseteq \phi(A) \setminus \{a\} \subseteq \phi(A)$.
    In particular $\rho(E)=\rho(j)+1$.
    Furthermore since $A$ is minimal generator of $j$, $E \not\subseteq \phi(A \setminus \{a\})$.
    Consequently, there is an element $e \in E$ which is implied by some $A' \subset A$ such that $a \in A'$.
    Hence $a$ and $e$ must have two different ranks.
    However since $a$ belongs to $\pred(j)$, it belongs to a critical minimal generator of $j$ and has rank $\rho(j)+1$.
    This contradicts $\rho(a)\neq \rho(e)$, and the theorem follows.
\end{proof}

We are now ready to describe an algorithm which enumerates the critical base of a ranked convex geometry given by the set of its meet-irreducible elements.
For every $j\in X$, it computes the set $\pred(j)$ of elements belonging to some critical minimal generator of~$j$, and the complementary hypergraph $\H_j$ of meet-irreducible elements in $\nearrow$ relation with~$j$.
Then according to Theorem~\ref{thm:mintr-characterization} it enumerates every minimal transversal of $\H_j[\pred(j)]$.
A trace of this algorithm is given in Example~\ref{ex:mintr-meet}.
Its complexity analysis is given in Theorem~\ref{thm:DFS-tr}.

\begin{example}\label{ex:mintr-meet}
    Consider the implicational base $\Sigma$ of Figure~\ref{fig:partition}.
    Note that $j^{\nearrow} = \{\{1, 2, 4\}$, $\{1, 3, 5\}, \{2, 3, 5\}\}$.
    We aim to compute the critical minimal generators of $j$.
    Here, $\E(\H_j) = \{\{3, 5\}, \{2, 4\}, \{1, 4\}\}$.
    Minimal generators of $j$ are $\{1, 2, 3\}, \{1, 2, 5\}$, $\{3,4\}$ and $\{4, 5\}$.
    Observe that for $1$, $\{2, 3, 5\} \cup \{1, j\}$ is not closed.
    The same goes for $2$ with $\{1, 3, 5\}$, and $3$ with $\{1, 2, 4\}$.
    On the other hand, $\{1, 2, 4\} \cup \{5, j\}$ is closed and $\{2, 3, 5\} \cup \{4, j\}$ too.
    Hence by Lemma~\ref{lem:pred}, one has $\pred(j) = \{4, 5\}$, and $\E(\H_j[\pred(j)]) = \{\{4\}, \{5\}\}$.
    It has a unique minimal transversal which is $\{4, 5\}$. 
    By Theorem~\ref{thm:mintr-characterization}, $\{4,5\}$ is a critical minimal generator of $j$. 
    Hence $45 \rightarrow j$ belongs to $\Sigma$.
    This coincides with $\Sigma$.
\end{example}

\begin{theorem}\label{thm:DFS-tr}
    Let $f,g:\mathbb{N}\rightarrow\mathbb{N}$ be two functions and $(X,\phi)$ be a ranked convex geometry of closure system $\C_\phi$ given by the set $\M=\M(\C_\phi)$ of its meet-irreducible elements. 
    Then there is an algorithm enumerating the ranked implicational base $(X,\Sigma)$ of critical minimal generators of $(X,\phi)$ with delay 
    \[
        O(f(|X|+|\M|)),
    \]
    $O(|X|\cdot |\M| + g(|X|+|\M|))$ space, and after $O(|X|^3\cdot |\M|^2)$ preprocessing time
    if there is one enumerating $Tr(\H)$ with delay $f(|\H|)$ and space $g(|\H|)$ given a hypergraph $\H$.
\end{theorem}

\begin{proof}
    By Theorem~\ref{thm:mintr-characterization}, $A\subseteq X$ is a critical minimal generator of $j\in X$ if and only if it is a minimal transversal of the hypergraph $\H_j[\pred(j)]$.
    Note that the closure $\phi(S)$ of a set $S\subseteq X$ can be computed in $O(|X|\cdot |\M|)$ time by intersecting every $M\in \M$ such that $S\subseteq M$.
    Given $M\in \M$, computing $j\in X$ such that $M\in j^\nearrow$ can be done in $O(|\M|\cdot |X|^2)$ time by checking whether $M\cup \{j\}$ is closed for every $j\in X\setminus M$.
    Hence, computing the partition $\{j^\nearrow \mid j\in X\}$ of $\M$ can be done in $O(|\M|^2\cdot |X|^2)$ time and $O(|X|\cdot |\M|)$ space.
    Using Lemma~\ref{lem:pred}, the sets $\pred(j)$, $j\in X$ can be computed in $O(|X|^3\cdot |\M|^2)$ time and $O(|X|^2)$ space by checking for every $j\in X$ and $a\in X$ whether there exists $M\in j^\nearrow$ such that $a\not\in M$ and $M\cup \{a,j\}$ is closed.
    Note that $O(|X|^2)$ is bounded by $O(|X|\cdot |\M|)$ as $\{j^\nearrow \mid j\in X\}$ is a partition of $\M$.
    In addition, the hypergraphs $\H_j[\pred(j)]$, $j\in X$ can be computed in $O(|X|^2\cdot |\M|)$ time and $O(|X|\cdot |\M|)$ space as 
    \[
      \sum_{j\in X} |\E(\H_j)|=|\M|.
    \]
    Ultimately, this computation can be achieved in $O(|X|^3\cdot |\M|^2)$ preprocessing time and using $O(|X|\cdot |\M|)$ space.
    
    Now by assumption and since $|\H_j[\pred(j)]| \leq |X|+|\M|$, the critical minimal generators of $j\in X$ can be enumerated with $f(|X|+|\M|)$ delay and using $g(|X|+|\M|)$ space.
    The theorem follows.
\end{proof}

A trivial bound of $f(|\H|+|Tr(\H)|)$ on the delay of an output-polynomial time algorithm enumerating $Tr(\H)$ within the same time yields an output quasi-polynomial time algorithm constructing the ranked implicational base of a ranked convex geometry given by the set of its meet-irreducible elements, using the algorithm of Fredman and Khachiyan~\cite{fredman1996complexity} as a subroutine for the enumeration of $Tr(\H)$.
If furthermore every meet-irreducible element $M\in \M(\C_\phi)$ is such that $|M|\geq |X|-k$ for some fixed integer $k\in \mathbb{N}$, then the described algorithm performs in output-polynomial time using the algorithm of Eiter and Gottlob in~\cite{eiter1995identifying}.

As in the previous section, observe that the characterization of Lemma~\ref{lem:pred} may go beyond the scope of ranked convex geometries. 
In particular, the algorithm of Theorem~\ref{thm:DFS-tr} will correctly output the critical base of the convex geometry given in Example~\ref{ex:open-pbs}.
This yields the next question.

\begin{problem}\label{op:sigma}
    For which closure spaces does Theorem~\ref{thm:mintr-characterization} hold?
\end{problem}

\section{Conclusion}\label{sec:conclusion}

In this paper, we investigated the complexity of translating between meet-irreducible elements and unit-minimum implicational bases in acyclic convex geometries.
Even when restricted to this class, we showed that the problem is harder than the dualization in distributive lattices, a generalization of the hypergraph dualization problem for which the existence of an output quasi-polynomial time algorithm is open.\footnote{We would like to mention that an output quasi-polynomial time algorithm for the dualization in distributive lattices has recently been announced by Elbassioni during the rewiewing process of this article \cite{elbassioni2020dualization}.}
Then, we considered a proper subclass of acyclic convex geometries, namely ranked convex geometries, as those that admit a ranked implicational base analogous to that of ranked posets. 
For this class, we provided output quasi-polynomial time algorithms based on hypergraph dualization for translating between the two representations.

Several questions arise for future research.
First, we would like to know how implications with premises of size one can be integrated to our result, regardless of the rank.
It seems that it only barely affects the strategy employed in Section~\ref{sec:mis} as long as these implications do not link the elements of a same rank.
This observation relates to Example~\ref{ex:open-pbs} and Open problems~\ref{op:meet} and~\ref{op:sigma}.
Note however that in general acyclic convex geometries, \textsc{CMI} is harder than the dualization in distributive lattices (c.f.~Theorem~\ref{thm:dual-hard}), and that the computation of $j^\nearrow$ is even harder than the dualization in lattices given by acyclic implicational bases (c.f.~Theorem~\ref{thm:dual-hard} and Equality~\ref{eq:dual-hard-antichain}).

Also, we leave open the question whether ranked convex geometries can be identified in polynomial time from meet-irreducible elements.
As shown in Proposition~\ref{prop:conp}, this problem lies in co{\sf NP}.

\bibliographystyle{alpha}
\bibliography{main}

\end{document}